\documentclass{llncs}

\pdfoutput=1

\usepackage{amssymb}
\usepackage{eucal}
\usepackage{graphicx}
\usepackage{amsmath}
\usepackage{color}
\usepackage{stmaryrd}
\usepackage{rotating}
\usepackage{xspace}
\usepackage{eepic}
\usepackage{epsfig}
\usepackage{latexsym} 
\usepackage{exscale}
\usepackage{verbatim}

\newcommand{\jd}[1]{\color{black}{#1}}
\newcommand{\agents}{A}

\newcommand{\outc}{R}
\newcommand{\outrel}{\delta}

\newcommand{\vx}{\vec{x}}
\newcommand{\next}{{\cal X}}

\newcommand{\trf}[1]{\langle #1 \rangle}
\newcommand{\N}{{\cal N}}

\newcommand{\je}[1]{\textcolor{black}{#1}}

\def\by#1{\mathop{{\hbox{\setbox0=\hbox{$\scriptstyle{#1\quad}$}{$\buildrel{\quad\scriptstyle{#1}\quad}\over{\hbox to \wd0{\rightarrowfill}}$}}}}}

\def\bytwo#1#2{%
\setbox0=\hbox{$\scriptstyle{#1\quad}$}%
\mathrel{\mathop{\hbox to \wd0{\rightarrowfill}}\limits^{#1}_{#2}}%
}

\begin{document}

\pagestyle{headings}  

\author{Javier Esparza\inst{1} \and J\"org Desel\inst{2}} 

\title{On Negotiation as Concurrency Primitive}
\institute{Fakult\"at f\"{u}r Informatik, Technische Universit\"{a}t M\"{u}nchen, Germany \and
Fakult\"at f\"ur Mathematik und Informatik, FernUniversit\"at in Hagen, Germany}
\maketitle

\begin{abstract}
We introduce negotiations, a model of concurrency close to Petri nets, with multiparty
negotiation as primitive. We study the problems {\jd of soundness of negotiations and} of, given a negotiation with possibly many steps, 
computing a {\em summary}, i.e., an equivalent one-step negotiation. We provide a complete set 
of reduction rules for \jd{sound}, \je{acyclic, weakly deterministic negotiations  and show that, for  
deterministic negotiations, the rules compute
the summary in polynomial time}.

\end{abstract}

\section{Introduction}

Many modern distributed systems consist of components whose behavior is 
only partially known. Typical examples include open systems where programs 
(e.g. Java applets) can enter or leave, multi-agent systems, business processes,
or protocols for conducting elections and 
auctions. 

An interaction between {\jd a fixed set of components with not 
fully known behavior} can be abstractly described as a {\em negotiation} in which several
{\em parties} (the components {\jd involved in the negotiation}) nondeterministically agree on an {\em outcome}, 
which results in a transformation of  internal states of the parties.
A more technical but less suggestive term
would be a {\em synchronized nondeterministic choice} and, as the name suggests,
these interactions can be modelled in any standard process algebra
as a combination of parallel composition and nondeterministic 
choice, or as small Petri nets. We argue that much can be gained by studying 
formal models with {\em negotiation atoms} as 
concurrency primitive. In 
particular, we show that the negotiation point of view reveals new classes of systems 
with polynomial analysis algorithms.

Negotiation atoms can be combined into {\em distributed negotiations}. For instance, a distributed negotiation between a buyer, a seller, and a broker,
consists of one or more rounds of atoms  involving the buyer and 
the broker or the seller and the broker, followed by a final atom 
between the buyer and the seller. 
We introduce a formal model for distributed negotiations,
close to a colored version of van der {\jd Aalst's} {\em workflow nets} \cite{aalst},
and investigate two {\jd important} analysis problems.
First, just like workflow nets, distributed negotiations can be {\em unsound} 
because of deadlocks 
or livelocks (states from which no deadlock is reached, but the negotiation 
cannot be completed). The {\em soundness} problem consists of deciding if
a given negotiation is sound. Second, a sound negotiation
is equivalent to a  negotiation with only one atom whose 
{\jd state transformation function} 
determines the possible final internal states of all parties
as a function of their initial internal states. We call this  negotiation a 
{\em summary}. The {\em summarization problem} consists of computing 
a summary of a distributed negotiation. 
Both problems will shown to be  PSPACE-hard for arbitrary negotiations, and NP-hard for acyclic ones. 
They can be solved by means of well-known algorithms based on the exhaustive 
exploration of the state space. However, this approach badly suffers from the state-explosion problem: 
even the analysis of distributed negotiations with a very simple structure requires exponential  time. 

In this paper we suggest {\em reduction} algorithms that 
avoid the construction of the state space but  exhaustively 
apply syntactic reduction rules 
that simplify the system while preserving some aspects of the behavior,
like absence of deadlocks. This approach has been extensively applied to 
Petri nets or workflow nets, but most of this work has been devoted to the 
liveness or soundness problems \cite{DBLP:conf/ac/Berthelot86,DBLP:conf/apn/Haddad88,DBLP:journals/ppl/HaddadP06,DBLP:conf/caise/DongenAV05,DBLP:journals/jcss/VerbeekWAH10}. 
For these problems many reduction rules are known, and some sets of rules have been proved
{\em complete} for certain classes of systems \cite{DBLP:journals/tcs/GenrichT84,DBLP:conf/concur/Desel90,Desel:1995:FCP:207572}, meaning that they reduce all live or sound systems in the class, and only those, to a trivial system (in our case to a single negotiation atom). However, many 
of these rules, like the linear dependency rule of \cite{Desel:1995:FCP:207572}, cannot be applied to  the summarization problem, because they  preserve  only the soundness property.
 
We present a complete set of reduction rules for the summarization problem of {\em acyclic} negotiations that are either {\em deterministic} or {\em weakly deterministic}. The rules are inspired by reduction rules used to transform finite automata into regular expressions by eliminating states \cite{HUM}. In deterministic negotiations all involved agents are deterministic, meaning that they are never ready to engage in more than one 
negotiation atom. Intuitively, 
nondeterministic agents may be ready to engage in multiple atoms, and which one takes place 
is decided by the deterministic parties, which play thus the role of negotiation
leaders.  In weakly deterministic negotiations not every agent is deterministic, but some deterministic party is involved in every negotiation atom an agent can engage in next. 

For {\em deterministic} negotiations we prove that 
a sound and acyclic negotiation can be summarized by means of a polynomial number of application of the rules, leading to a polynomial algorithm.

The paper is organized as follows. Section \ref{sec:synsem} introduces the syntax and semantics of the model. 
Section \ref{sec:problems} introduces the
soundness and summarization problems. Section \ref{sec:redrules} presents our reduction rules. Section \ref{sec:determ} defines (weakly) deterministic   negotiations.
Section~\ref{sec:compcomp} proves the completeness and polynomial complexity results announced above. 
Finally, Section \ref{sec:conc} presents some conclusions, open questions and related work.

This paper is an extended version of the conference paper \cite{Concur2013}.

\section{Negotiations: Syntax and Semantics}
\label{sec:synsem}

We fix a finite set $\agents$ of {\em agents} representing potential parties of negotiations.
Each agent $a \in \agents$ has a (possibly infinite) nonempty set $Q_a$ of {\em internal states}. 
We denote by $Q_\agents$ the cartesian product $\prod_{a \in \agents} Q_a$. 
{\jd A {\em transformer} is a left-total relation $\tau \subseteq Q_\agents \times Q_\agents$}, representing a nondeterministic state transforming function.
{\jd Given $S \subseteq \agents$,
we say that a transformer $\tau$ is an {\em $S$-transformer} if, for each $a_i \notin S$,  
$\left((q_{a_1} , \ldots , q_{a_i}, \ldots , q_{a_{|A|}}), (q'_{a_1} , \ldots ,q'_{a_i}, \ldots ,  q'_{a_{|A|}})\right)\in\tau$ implies $q_{a_i} = q'_{a_i}$.}
So an $S$-transformer  only transforms the internal
states of  agents in $S$. 

\begin{definition}
A {\em negotiation atom}, or just an {\em atom}, is a triple $n=(P_n, \outc_n,\outrel_n)$,
 where $P_n \subseteq \agents$ is a nonempty set of {\em parties}, 
$\outc_n$ is a finite, nonempty set of {\em outcomes}, {\jd and 
 $\outrel_n$  is a mapping assigning to each outcome $r$} in $\outc_n$ 
a $P_n$-transformer $\outrel_n (r)$.
We denote the transformer 
$\outrel_n (r)$ by $\trf{n,r}$, and, if there is no confusion, by $\trf{r}$.
\end{definition}

\noindent Intuitively, if the  
states of the agents before a negotiation $n$ are given by a tuple $q$ 
and the outcome of the negotiation is $r$, then the agents  change
their  states to $q'$
for some $(q,q') \in \trf{n,r}$. 
Only the parties of $n$ can change their internal states. Each outcome $r \in \outc_n$ is possible, independent from the previous internal states of the parties.

For a simple example, consider a negotiation atom $n_\texttt{FD}$ with 
parties \texttt{F} (Father) and \texttt{D} (teenage Daughter). The goal of the
negotiation is to determine whether \texttt{D} can go to a party, and the time 
at which she must return home. The possible outcomes are
$\{\texttt{yes}, \texttt{no}, \texttt{ask\_mother}\}$.
Both sets $Q_{\texttt{F}}$ and $Q_{\texttt{D}}$ contain a state 
{\it angry} plus a state $t$ for every time $T_1 \leq t \leq T_2$ in a 
given interval $[T_1,T_2]$. {\jd Initially, \texttt{F} is in state $t_f$ and \texttt{D} in state $t_d$.} 
The  transformer $\delta_{n_\texttt{FD}}$ is given by

$$
\begin{array}{rcl}
\trf{\texttt{yes}}  & = & \left\{ \left( (t_f, t_d)  , (t,t)\right) \; \mid \;  t_f \leq t \leq t_d \vee t_d \leq t \leq t_f \right\} \\
\trf{\texttt{no}} & = &  \left\{\left((t_f, t_d) ,  ({\it angry}, {\it angry})\right) \; \right\}  \\
\trf{\texttt{ask\_mother}}& = & \left\{ \left( (t_f, t_d) ,  (t_f, t_d) \right) \right\} 
\end{array} 
$$
\noindent That is, if the outcome is \texttt{yes}, then \texttt{F} and \texttt{D} agree on a time $t$ {\jd which is not earlier and not later than both suggested times}.
If it is \texttt{no}, then there is a quarrel and both parties get angry. If the
outcome is \texttt{ask\_mother}, then the parties keep their previous times. 

\subsection{Combining negotiation atoms}

If the outcome of the atom above is \texttt{ask\_mother}, then $n_\texttt{FD}$
should be followed by a second atom $n_\texttt{DM}$ between \texttt{D} and \texttt{M} 
(Mother). The complete negotiation is the composition of $n_\texttt{FD}$ and $n_\texttt{DM}$, 
where the possible internal states of \texttt{M} are the same as those of \texttt{F} and \texttt{D}, and 
$n_\texttt{DM}$ is  a 
copy of $n_\texttt{FD}$, but without the \texttt{ask\_mother} outcome. In 
order to compose atoms, we add a {\em transition function} $\next$ that assigns to every 
triple $(n,a,r)$ consisting of an atom $n$, a party $a$ of $n$, and an outcome $r$ of $n$ a set 
$\next(n,a,r)$ of atoms. Intuitively, this is the set of atoms  
agent $a$ is ready to engage in after the atom $n$,  if the outcome of $n$ is $r$. 

\begin{definition}
Given a {\jd finite} set of atoms $N$, let  $T(N)$ denote the set of triples $(n, a, r)$ such that 
$n \in N$, $a\in P_n$, and $r \in \outc_n$. 
A {\em negotiation} is a tuple ${\cal N}=(N, n_0, n_f, \next)$, where 
$n_0, n_f \in N$ are the {\em initial} and {\em final} atoms, and  $\next \colon T(N) \rightarrow 2^N$ is the {\em transition function}. Further, ${\cal N}$ satisfies the following properties: 

(1) every agent of $\agents$ participates in both $n_0$ and $n_f$; 

(2) for every $(n, a, r) \in T(N)$: $\next(n, a, r)= \emptyset$ if{}f $n=n_f$.

\end{definition}

{\jd We may have $n_0=n_f$. Notice that $n_f$ has, as all other atoms, at least one outcome $r \in \outc_{n_f}$.} 

\begin{figure}[t]
\centerline{\scalebox{0.35}{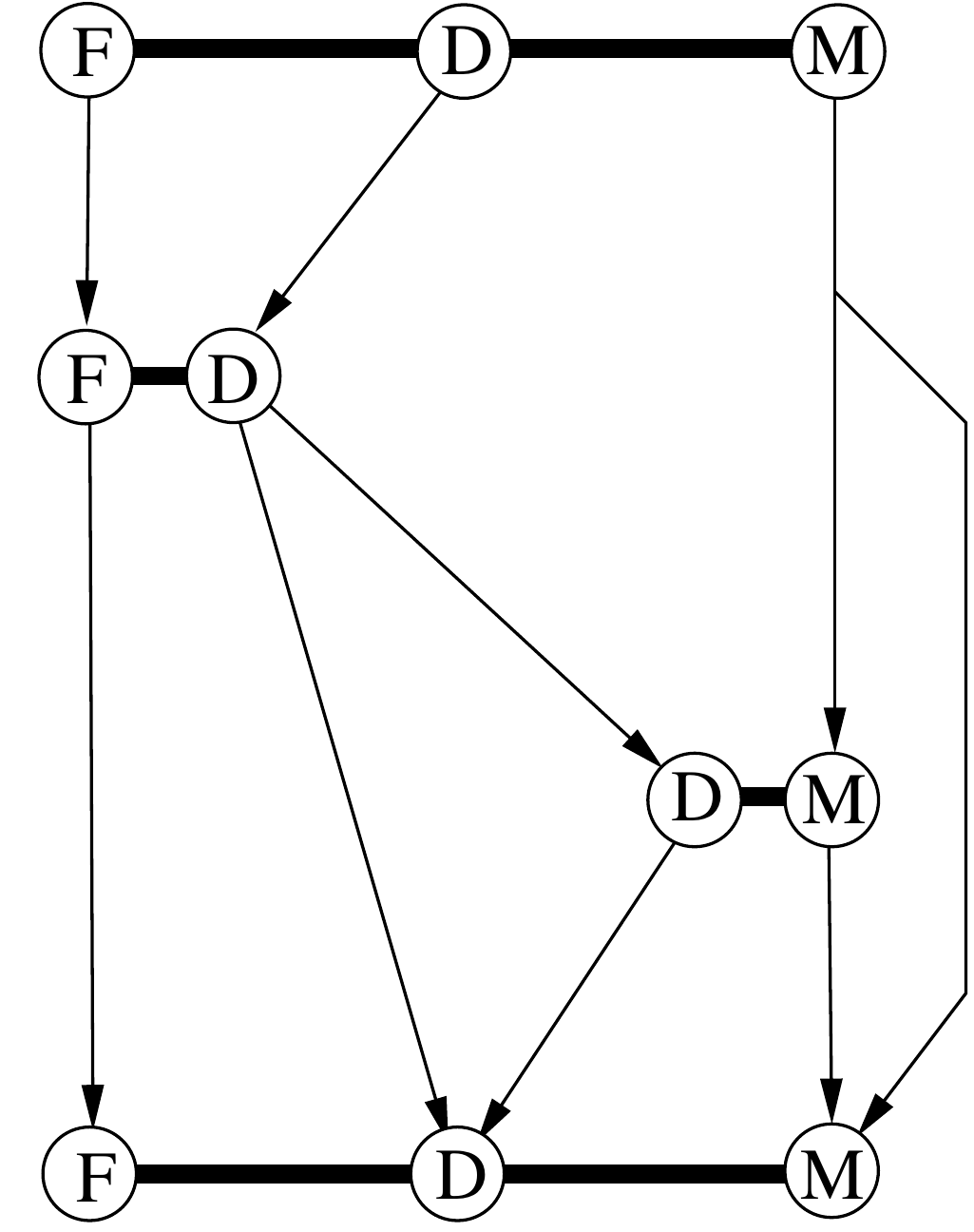} \qquad \scalebox{0.32}{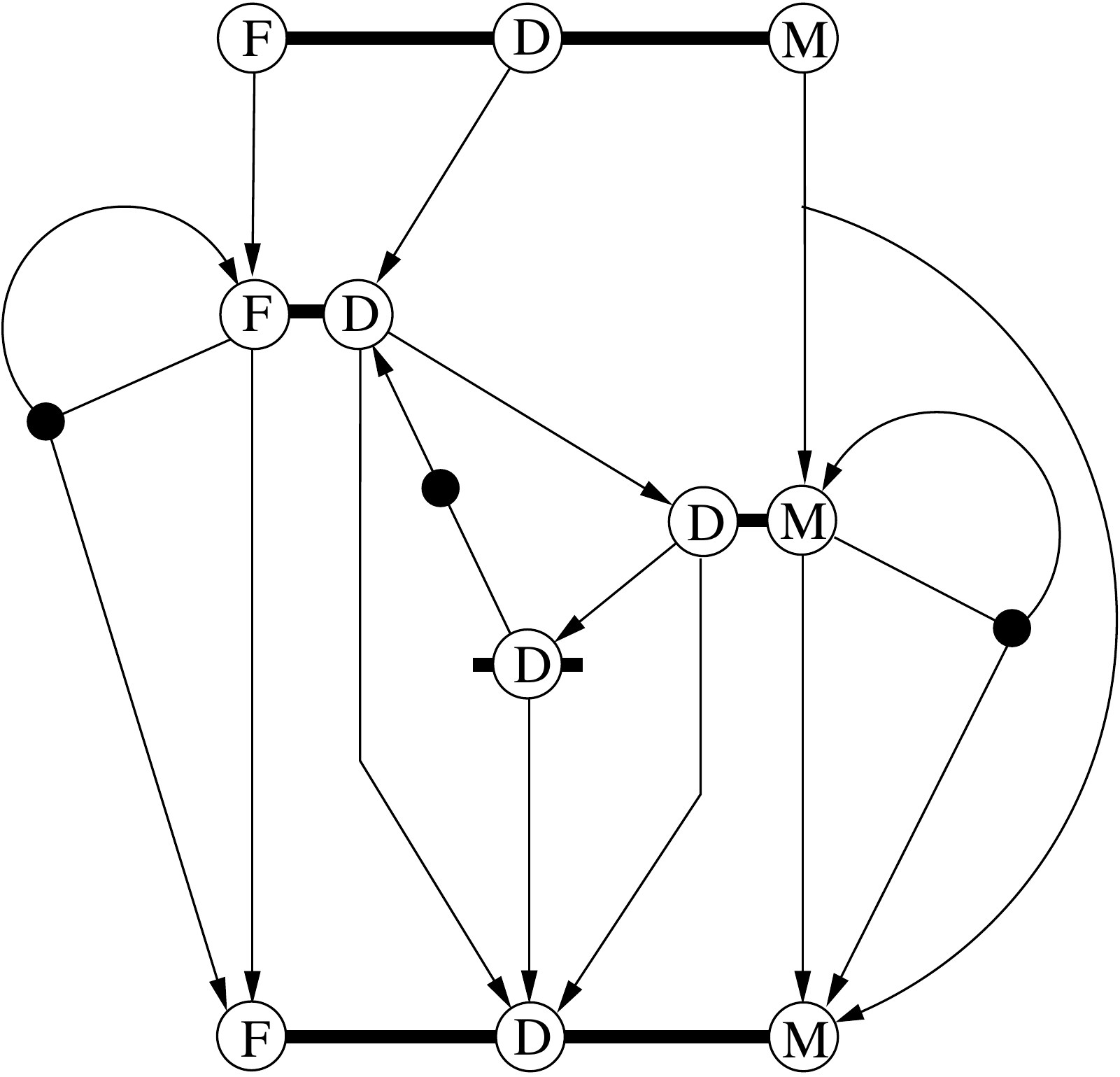}}
\caption{An acyclic negotiation and the ping-pong negotiation.}
\label{fig:dneg}
\end{figure}
Negotiations are graphically represented as shown in Figure \ref{fig:dneg}. For each atom $n \in N$ we draw a black bar; for each party $a$ of $P_n$ we draw a white circle on the bar, called a {\em port}. For each $(n,a,r) \in T(N)$, we draw a hyperarc leading from the port of $a$ in $n$ to all the ports of $a$ in the atoms of $\next(n,a,r)$, and label it by $r$. Figure \ref{fig:dneg} shows on the left the graphical representation of the Father-Daughter-Mother 
negotiation sketched above.
Instead of multiple (hyper)arcs connecting the same input port to the same output ports we draw a single (hyper)arc with multiple labels. In the figure, we write \texttt{y} for \texttt{yes}, \texttt{n} for \texttt{no}, and \texttt{am} for \texttt{ask}\underline{ }\texttt{mother}. \texttt{st} stands for \texttt{start}, the only outcome of  $n_0$.
Since $n_f$ has no outgoing arc, the outcomes of $n_f$ do not appear in the graphical representation.

\begin{definition}
The {\em graph associated to a negotiation ${\cal N}=(N, n_0, n_f, \next)$}
is the directed graph with vertices $N$ and edges $\{(n,n')\in N \times N \mid \exists \, (n,a,r) \in T(N) \colon n' \in \next(n,a,r)\} $. The negotiation ${\cal N}$ is {\em acyclic} if its graph has no cycles.
\end{definition}  
 
The negotiation on the left of Figure \ref{fig:dneg} is acyclic. The negotiation on the right
(ignore the black dots on the arcs for the moment) is the ping-pong negotiation, well-known in every family.
The $n_\texttt{DM}$ atom has now an extra outcome \texttt{ask\_father} (\texttt{af}), and Daughter
\texttt{D} can be sent back and forth between Mother and Father. After each round, \texttt{D} 
``negotiates with herself'' (atom $n_\texttt{D}$) with possible outcomes \texttt{c} (\texttt{continue}) and \texttt{gu} (\texttt{give up}). This negotiation is cyclic because, for instance, we have 
$\next(n_\texttt{FD}, \texttt{D}, \texttt{am}) = \{n_\texttt{DM}\}$, $\next(n_\texttt{DM}, \texttt{D}, \texttt{af}) = \{n_\texttt{D}\}$, and $\next(n_\texttt{D}, \texttt{D}, \texttt{c}) = \{n_\texttt{FD}\}$.

\subsection{Semantics}
A {\em marking} of a negotiation ${\cal N}=(N, n_0, n_f, \next)$ is a mapping 
$\vx \colon \agents \rightarrow 2^N$. Intuitively, $\vx(a)$ is the set of atoms that agent $a$ is currently ready to engage in next. 
The {\em initial} and {\em final} markings, 
denoted by $\vx_0$ and $\vx_f$ respectively, are given by $\vx_0(a)=\{n_0\}$ and $\vx_f(a)=\emptyset$ for every $a \in \agents$.

A marking $\vx$ {\em enables} an atom $n$ if $n \in \vx(a)$ for every $a \in P_n$,
i.e., if every agent that parties in $n$ is currently ready to engage in it.
If $\vx$ enables $n$, then $n$ can take place and its parties
agree on an outcome $r$; we say that $(n,r)$ {\em occurs}.
The occurrence of $(n,r)$ produces a next marking $\vx'$ given by $\vx'(a) = \next(n,a,r)$ for every $a \in P_n$, 
and $\vx'(a)=\vx(a)$ for every $a \in \agents \setminus P_n$. 
We write $\vx \by{(n,r)} \vx'$ to denote this,
and call it a {\em small step}. 

{\jd By this definition, $\vx (a)$ is always either $\{n_0\}$ or equals $\next(n,a,r)$ for some atom $n$ and outcome $r$. The marking $\vx_f$ can only be reached by the occurrence of $(n_f, r)$ ($r$ being a possible outcome of $n_f$), and it does not enable any atom. Any other marking that does not enable any atom is considered a {\em deadlock}.}

Reachable markings can be graphically represented by placing tokens {\jd (black dots)} on the forking points of the hyperarcs (or in the middle of an arc). {\jd Thus, both the initial marking and the final marking are represented by no tokens, and all other reachable markings are represented by exactly one token per agent.} 

Figure \ref{fig:dneg} shows on the right
the marking in which Father (\texttt{F}) is ready to engage in the atoms $n_\texttt{FD}$ and  $n_f$, 
Daughter (\texttt{D}) is only ready to engage in $n_\texttt{FD}$, and Mother (\texttt{M}) is ready to engage in both $n_\texttt{DM}$ and $n_f$. 

We write $\vx_1 \by{\sigma}$ to denote that there is a sequence 
$$\vx_1 \by{(n_1,r_1)} \vx_2 \by{(n_2,r_2)}\cdots \by{(n_{k-1},r_{k-1})} \vx_{k} \by{(n_k,r_k)} \vx_{k+1} \cdots$$ 
of small steps such that  $\sigma = (n_1, r_1) \ldots (n_{k}, r_{k}) \ldots$. If
$\vx_1 \by{\sigma}$, then $\sigma$ is an {\em occurrence sequence} from the marking $\vx_1$, and $\vx_1$ enables $\sigma$.
If $\sigma$ is finite, then we write
$\vx_1 \by{\sigma} \vx_{k+1}$ and say that $\vx_{k+1}$ is {\em reachable} from $\vx_1$. 
If  $\vx_1$ is the initial marking then we call $\sigma$ {\em initial occurrence sequence}. If moreover $\vx_{k+1}$ is the final marking, then $\sigma$ is a {\em large step}.

\section{Analysis Problems}
\label{sec:problems}

Correct negotiations should be deadlock-free and, in principle,  
they should not have infinite occurrence sequences either. However,
requiring the latter in our negotiation model is too strong, because
infinite occurrence sequences may be excluded by fairness constraints. 
Following \cite{aalst,DBLP:journals/fac/AalstHHSVVW11}, we introduce a notion of partial correctness independent of termination: 

\begin{definition}
A negotiation is {\em sound} if {\em (a)} every atom is enabled at some reachable marking, and {\em (b)} every occurrence sequence from the initial marking is either a large step or can be extended to a large step. 
\end{definition}

The negotiations of Figure \ref{fig:dneg} are sound. However, if we set
$\next(n_0,\texttt{M}, \texttt{st})= \{n_\texttt{DM}\}$ instead of $\next(n_0,\texttt{M}, \texttt{st})= \{n_\texttt{DM}, n_f\}$, then the occurrence sequence $(n_0,\texttt{st}) (n_\texttt{FD}, \texttt{yes})$
leads to a deadlock.

The {\em final outcomes} of a negotiation are the outcomes of its final atom.
Intuitively, two sound negotiations over the same agents 
are equivalent if they have the same final outcomes, and for each final outcome
they transform the same initial states into the same final states.

\begin{definition}
Given a negotiation ${\cal N}=(N,n_0,n_f,\next)$, we attach to each outcome $r$ of $n_f$ a 
{\em summary transformer} $\trf{{\cal N},r}$ as follows. Let $E_r$ be the set of large steps 
of ${\cal N}$ that end with $(n_f,r)$. We define $\trf{{\cal N},r} = \bigcup_{\sigma \in E_r} \trf{\sigma}$, where for $\sigma = (n_1, r_1)\, (n_2, r_2)
\ldots (n_k, r_k)$ we define $\trf{\sigma}= \trf{n_1, r_1} \, \trf{n_2, r_2}
 \cdots \trf{n_k, r_k}$ {\jd (remember that each $\trf{n_i, r_i}$ is a relation on $Q_A$;  concatenation is the usual concatenation of relations)}.
\end{definition}

$\trf{{\cal N},r}(q_0)$ is the set of possible final states of the agents after the 
negotiation concludes with outcome $r$, if their initial states are given by $q_0$. 

\begin{definition}
Two negotiations ${\cal N}_1$ and ${\cal N}_2$ over the same set of agents are {\em equivalent},
denoted by ${\cal N}_1 \equiv {\cal N}_2$, if they are either both unsound, or if they are both sound, 
have the same final outcomes, and 
$\trf{{\cal N}_1, r} = \trf{{\cal N}_2, r}$ for every final outcome $r$.
If ${\cal N}_1 \equiv {\cal N}_2$ and ${\cal N}_2$ {\jd consists of a single} atom, then ${\cal N}_2$ 
is a {\em summary} of ${\cal N}_1$. 
\end{definition}

Notice that, according to this definition, all unsound negotiations are equivalent. This amounts
to considering soundness essential for a negotiation: if it fails, we do not care about the rest. 

\subsection{Deciding soundness}

The {\em reachability graph} of a negotiation ${\cal N}$ has all markings reachable from $\vx_0$ as {\jd vertices}, and an arc leading from $\vx$ to $\vx'$ whenever $\vx \by{(n,r)} \vx'$.

The soundness problem consists of deciding if a given negotiation is
sound. It can be solved by (1) computing the reachability graph of ${\cal N}$
and (2a) checking that every atom appears at some arc, and (2b) that, for every reachable marking $\vx$, there is an occurrence sequence $\sigma$
such that $\vx \by{\sigma} \vx_f$.

Step (1) needs exponential time, and steps (2a) and (2b) are polynomial in the size of the reachability graph. 
So the algorithm is single exponential in 
the number of atoms. This cannot be easily avoided, because the problem is PSPACE-complete, and NP-complete for acyclic negotiations.

Recall that a language $L$ is in the class DP if there exist languages $L_1, L_2$ in NP and
co-NP, respectively, such that $L=L_1 \cap L_2$ \cite{PapadimitriouY82}.

\begin{theorem}
The soundness problem is PSPACE-complete. For acyclic negotiations, the problem is 
co-NP-hard and in DP (and so at level $\Delta^P_2$ of the polynomial hierarchy).
\end{theorem}
\begin{proof}

\noindent
{\em The soundness problem is in PSPACE.} 

\noindent
Membership in PSPACE can be proved by observing that the soundness problem can be formulated in CTL, and then applying the PSPACE algorithm for CTL and 1-safe Petri nets of \cite{DBLP:conf/ac/Esparza96}. This algorithm only assumes that, given a marking, one can compute a successor marking in polynomial time, which is the case both for Petri nets and for negotiations. However, since we only need a very special case of the CTL algorithm, we provide a self-contained proof.

We show that both conditions for soundness can be checked in nondeterministic polynomial space. The result then follows from Savitch's theorem (NPSPACE=PSPACE).

The first condition is: every atom is enabled at some reachable
marking. For this we consider each atom $n$ in turn, and guess step
by step an occurrence sequence ending with an occurrence of $n$.
This only requires to store the marking reached by the sequence executed so far.

The second condition is: every occurrence sequence from the initial
marking is either a large step or can be extended to a large step.
This case is a bit more involved. Let $S$ denote the problem of
checking this second condition. We prove $S \in \text{PSPACE}$.

\begin{itemize}
\item[(1)] The following problem is in PSPACE: given some
marking $\vx$, check that no occurrence sequence starting at $\vx$
ends with the final atom. \\
Let us call this problem NO-OCC. We have
$\overline{\text{NO-OCC}} \in \text{NPSPACE}$, because we can nondeterministically guess
an occurrence sequence starting at $\vx$ that
ends with the final atom (we guess one step at a time). Since
NPSPACE=PSPACE=co-PSPACE, we get $\text{NO-OCC} \in
\text{PSPACE}$.\\

\item[(2)] $\overline{S} \in \text{NPSPACE}$. \\
$\overline{S}$ consists of
checking the existence of a sequence $\sigma$, firable from the
initial marking, that is neither a large step nor can be extended to
it. For this we guess a sequence $\sigma$ step by step that does not
end with the final atom. Then we consider the marking $\vx$ reached by
the occurrence of $\sigma$. Clearly, we have $\sigma \in \overline{S}$
if{}f $\vx \in \text{NO-OCC}$. So it suffices to apply our
deterministic polynomial-space algorithm for NO-OCC (see (1)).\\

\item[(3)] $S \in PSPACE$.\\
Follows from (2) and NPSPACE=PSPACE=co-PSPACE.
\end{itemize}

\noindent
{\em The soundness problem is PSPACE-hard.} 

\noindent
For PSPACE-hardness,
we reduce the problem of deciding if a
deterministic linearly bounded automaton (DLBA) recognizes an input to
the soundness problem. Let $A=(Q,\Sigma,\delta,q_0,F)$ be a DLBA, and
consider an input $w=a_1 \ldots a_k \in \Sigma^*$. The construction is
very similar to that of \cite{DBLP:conf/ac/Esparza96} for proving PSPACE -hardness of the reachability problem for 1-safe Petri nets, and so we do
not provide all details. The negotiation ${\cal N}_A$ has a control
agent $C$, a head agent $H$, and a cell agent $T_i$ for every tape
cell (i.e., $1 \leq i \leq k$). All agents have only one internal
state, i.e., the internal states are irrelevant. The negotiation has
an atom $n[q,h,a]$ (with only one outcome) for every
state $q$, every head position $1 \leq h \leq k$, and every $a \in
\Sigma$, plus an initial atom $n_0$ and a final atom $n_f$. The
parties of $n[q,h,a]$ are $C$, $H$, and $T_h$. The transition function
$\next$ is defined so that ${\cal N}_A$ simulates $A$ in the following
sense: if $A$ is currently in state $q$ with the head at position $h$,
and the contents of the tape are $b_1 \ldots b_k$, then the current
marking $\vx$ of the negotiation satisfies the following properties:

\begin{itemize}
\item if $q\neq q_f$, then $\vx(C)$ is the set of atoms $n[h',q', a]$
such that $q'=q$, and both $h'$ and $a$ are arbitrary; if $q = q_f$,
then $\vx(C) = \{ n_f\}$;

\item $\vx(H)$ is the set of atoms $n[h',q', a]$ such that $h'=h$ and $q',a$ are arbitrary, plus the final atom;

\item $\vx(T_i)$ is the set of atoms $n[h',q', a]$ such that $h'=i$, $q'$ is arbitrary, and $a=b_i$, plus the final atom.
\end{itemize}

\noindent (Intuitively, agent $C$ is only ready to
engage in atoms for the state $q$; agent $H$ is only ready to engage
in atoms for the position $h$; and $T_h$ is only ready to engage in
atoms for the letter $b_h$.) These properties guarantee that the only atom enabled by $\vx$ is
$n[h,q,b_h]$ if $q \neq q_f$, or the atom $n_f$ if $q=q_f$. So the
negotiation ${\cal N}_A$ has only one initial occurrence sequence, which
corresponds to the execution of $A$ on $w$.

It remains to define $\next$ so that it satisfies these properties. For the initial atom we take (recall that the input of the DLBA $A$ is the word $w=a_1 \ldots a_k$):

$$\begin{array}{rcl}
\next(n_0, C, {\it step}) & = & \{n[h',q_0, a'] \mid 1 \leq h' \leq k, a' \in \Sigma \} \\
\next(n_0, H, {\it step}) & = & \{n[1,q', a'] \mid q' \in Q, a' \in \Sigma \} \\
\next(n_0, T_{i}, {\it step}) & = & \{n[i,q_0, a_i] \}
\end{array}$$

For the transition function of an atom $n[q,h,a]$ we must consider the three possible cases of the transition relation
(head moves to the right, to the left, or stays put). We only deal with the case in which the machine moves to the
right, the others being analogous. Assume $\delta(q,a)=(\hat{q},\hat{a},R)$.
Then we take
$$\begin{array}{rcl}
\next(n[h,q,a], C, {\it step}) & = & \{n[h',\hat{q}, a'] \mid 1 \leq h' \leq k, a' \in \Sigma \} \\
\next(n[h,q,a], H, {\it step}) & = & \{n[h+1,q', a'] \mid q' \in Q, a' \in \Sigma \} \\
\next(n[h,q,a], T_{h}, {\it step}) & = & \{n[h,q', \hat{a}] \mid q' \in Q \}
\end{array}$$

Since $A$ is deterministic, ${\cal N}_A$ has only one maximal occurrence sequence, which is
a large step if{}f $A$ accepts. So ${\cal N}_A$ is sound if{}f $A$ accepts.\\

\noindent
{\em The soundness problem for acyclic negotiations is in DP.}

\noindent
We first observe that no occurrence of an acyclic negotiation contains an atom more than once
(loosely speaking, once the tokens of the parties of the atom have ``passed'' beyond it, they cannot return). It follows that the length of an occurrence sequences is at most equal to the number of atoms. To check soundness we must
check that (1) every atom can be enabled, and that (2) every occurrence sequence can be extended to a large step.
Checking (1) can be done by guessing
in polynomial time enabling sequences for all atoms, and so (1) is in NP.
Checking the negation of (2) can be done by guessing in polynomial
time an occurrence sequence that cannot be extended to a large step,
and so (2) is in coNP. So the conjunction of (1) and (2) is in DP.\\

\noindent
{\em The soundness problem for acyclic negotiations is co-NP-hard.}

\noindent We reduce 3-CNF-SAT to non-hardness.
Given a boolean formula $\phi$ with variables
$x_i$, $1 \leq i \leq n$ and clauses $c_j$, $1 \leq j \leq m$, we construct a
negotiation ${\cal N}_\phi$ with an agent $X_i$ for each $x_i$, and
an agent $J$ (for judge). W.l.o.g. we assume that no clause of $\phi$ is a tautology.
For each variable $x_i$, ${\cal N}_\phi$ has an atom
${\it Set\_x}_i$ with $X_i$ as only party and outcomes $\texttt{true}$ and $\texttt{false}$.
For each clause $c_j$, the negotiation ${\cal N}_\phi$ has an atom
${\it False}_j$ whose parties are the variables appearing in $c_j$ and the judge $J$.
The atom has only one outcome $\texttt{false}$.

After the initial atom, agent $X_i$ engages in ${\it Set\_x}_i$ and sets $x_i$ to a value
$b \in \{\texttt{true},\texttt{false}\}$ by choosing the appropriate outcome. After that, $X_i$
is ready to engage in the atoms ${\it False}_j$ satisfying the following condition:
the clause $c_j$ is {\em not} made true by setting $x_i$ to $b$; moreover, it
is also ready to engage in the final atom. As a consequence, ${\it False}_j$ becomes enabled if{}f the assignment chosen by the $X_i$'s makes $c_j$ false.
Finally, after the occurrence of a ${\it False}_j$, its parties are only ready to engage in the final atom.

After the initial atom, the judge $J$ is ready to engage in all atoms ${\it False}_j$, and then, if any of them occurs, in the final atom.

We argue that ${\cal N}_\phi$ is sound if{}f $\phi$ is satisfiable.
Notice first that, since by assumption no clause is
a tautology, every ${\it False}_j$ atom is enabled by some occurrence sequence.
So all atoms but perhaps the final atom can be enabled by some sequence. So
${\cal N}_i$ is sound if{}f every occurrence sequence can be extended to a large step,
and therefore it suffices to show that $\phi$ is satisfiable if{}f
every occurrence sequence of ${\cal N}_\phi$ can be extended to a large step.

If $\phi$ is unsatisfiable then, whatever the assignment determined
by the outcomes of the ${\it Set\_x}_i$'s, some clause is
false, and so at least one of the ${\it False}_j$ atoms is enabled. After
some ${\it False}_j$ occurs, the final atom becomes enabled, and so
the computation can be extended to a large step.

If $\phi$ is satisfiable, then consider an initial occurrence sequence in which the atoms ${\it Set\_x}_i$ occur, and they choose the outcomes corresponding to a satisfying assignment. Then none of the ${\it False}_j$ atoms become enabled.
Moreover the final atom is not enabled either, because the judge $J$ is not ready to engage
in it. So the occurrence sequence cannot be extended to a large step.
\qed
\end{proof}

\subsection{A summarization algorithm}
\label{subsec:sumalg}
The {\em summarization problem} consists of computing
a summary of a given negotiation, if it is sound. 
A straightforward solution follows these steps:

\begin{itemize}
\item[(1)] Compute the reachability graph of ${\cal N}$. Interpret it as 
a weighted finite automaton over the alphabet of transformers $\trf{n,r}$,  
with $\vx_0$ as initial state, and $\vx_f$ as final state.
\item[(2)] Compute the sum over all paths $\sigma$ leading from $\vx_0$ to $\vx_f$ of the
transformers $\trf{\sigma}$. We recall a well-known algorithm for this based on state
elimination (see e.g. \cite{HUM}). The algorithm proceeds in phases consisting of 
the following three steps:
\begin{itemize}
\item[(2.1)] Exhaustively replace steps $\vx \by{f_1} \vx'$, $\vx \by{f_2} \vx'$
by one step $\vx \by{f_1+f_2} \vx'$.
\item[(2.2)] Pick a state $\vx$ different from $\vx_0$ and $\vx_f$. If there is a 
self-loop $\vx \by{f} \vx$, replace all steps $\vx \by{g} \vx'$, where
$\vx' \neq \vx$, by the step $\vx \by{g^* f} \vx'$, and then remove the self-loop.
\item[(2.3)] For every two steps $\vx_1 \by{f_1} \vx$ and $\vx \by{f_2} \vx_2$, 
add a step $\vx \by{f_1 f_2} \vx_2$ and remove state $\vx$ together with its
incident steps.
\end{itemize}
\end{itemize}

Clearly, step  (1) takes exponential time in the number of atoms. Steps
 (2.1)-(2.3) can be seen as {\em reduction rules} that replace an automaton by a smaller one
with the same sum over all paths. In the next section we provide similar rules, but at the
syntactic level, i.e., rules {\jd which} act directly on the negotiation diagram, and {\em not} on the
reachability graph. This avoids the construction of the reachability graph.
Two of the three rules are straightforward generalizations of (2.1) and (2.3)
above, while the third allows us to 
remove  certain 
useless arcs from a negotiation.

\section{Reduction Rules}

\label{sec:redrules}

A {\em reduction rule}, or just a rule, 
is a binary relation on the set of negotiations. Given a rule $R$,
we write ${\cal N}_1 \by{R} {\cal N}_2$ for $({\cal N}_1, {\cal N}_2) \in R$.
A rule $R$ is {\em correct} if it preserves equivalence, i.e., if 
${\cal N}_1 \by{R} {\cal N}_2$ implies ${\cal N}_1 \equiv{\cal N}_2$. 
Notice that, in particular, this implies that ${\cal N}_1$ is sound if and only if 
${\cal N}_2$ is sound.

Given a set of
rules ${\cal R} = \{R_1, \ldots, R_k\}$, we denote by ${\cal R}^*$ the reflexive 
and transitive closure of $R_1 \cup \ldots \cup R_k$. We say that ${\cal R}$
is {\em complete with respect to a class of negotiations} if, for every
negotiation ${\cal N}$ in the class, there is a negotiation ${\cal N'}$ consisting of a single atom
such that ${\cal N} \by{{\cal R}^*} {\cal N'}$.

We describe rules as pairs of a {\em guard} and an {\em action}; 
${\cal N}_1 \by{R} {\cal N}_2$ holds if ${\cal N}_1$ satisfies the guard and 
${\cal N}_2$ is a possible result of applying the action to ${\cal N}_1$.\\

\noindent
{\bf \em Merge rule.} Intuitively, the {\em merge rule} merges two outcomes with identical {\jd next enabled atoms} into one single outcome. It corresponds to the rule of step (2.1) in the previous section.

\begin{definition}{Merge rule}

\noindent {\bf Guard}: \begin{tabular}[t]{l} $N$ contains an atom $n$ 
with two distinct outcomes $r_1, r_2 \in \outc_n$,
 such \\ that $\next(n,a,r_1) = \next(n,a,r_2)$ for every $a \in \agents_n$.
\end{tabular}

\noindent {\bf Action}: \begin{tabular}[t]{ll}
(1) & $\outc_n \leftarrow (\outc_n \setminus \{r_1, r_2\}) \cup \{r_f\}$,
where $r_f$ is a fresh name. \\
(2) & For all $a \in P_n$: $\next(n,a,r_f) \leftarrow \next(n,a,r_1)$. \\
(3) & $\delta (n, r_f) \leftarrow \delta (n, r_1) \cup \delta (n, r_2)$.
\end{tabular}
\end{definition}
It is easy to see that the merge rule is correct for arbitrary negotiations.\\

\noindent {\bf \em Shortcut rule.} The shortcut rule corresponds to the rule of step (2.3) in the previous section.
We need a preliminary definition. 
\je{\begin{definition}
Given atoms $n,n'$, we say that $(n,r)$ {\em unconditionally enables} $n'$
if $P_n \supseteq P_{n'}$ and $\next(n,a,r) = \{n'\}$ for every $a \in P_{n'}$.
\end{definition}}
\noindent Observe that if $(n,r)$ unconditionally enables $n'$
then, for {\em every} marking $\vx$ that enables $n$, 
the marking $\vx'$ given by $\vx \by{(n,r)} \vx'$ enables $n'$. \je{Moreover, $n'$ 
can only be disabled by its own occurrence.}

Loosely speaking, the shortcut rule merges the outcomes of two atoms that can occur one after 
the other into one single outcome with the same effect. Consider the negotiation fragment shown 
on the left of
Figure \ref{fig:shortcut}. The guard of the rule will state that $n$ must
unconditionally enable $n'$, which is the case. For every outcome of $n'$, 
say $r_1$, the action of the rule adds  
a fresh outcome $r_{1f}$ to $n$, and modifies the negotiation 
so that the occurrence of $(n,r_{1f})$ has the same 
effect as the occurrence of the sequence $(n,r)(n',r_1)$. In the figure, 
shortcutting the outcome $(n,r)$ leaves $n'$ without any input arc, and in 
this case the rule also removes $n'$.
Otherwise we require that at least one input arc of a party $\tilde{a}$ of  $n'$ is an arc (i.e., not a proper hyperarc) from some atom $\tilde{n}\neq n$, annotated by $\tilde{r}$. This implies that after the occurrence of $(\tilde{n},\tilde{r})$, $n'$ is the only atom agent $\tilde{a}$ is ready to engage in.

\begin{definition}{Shortcut rule}
\label{def:shortcutrule}

\noindent {\bf Guard}: $N$ contains an atom $n$ with an outcome $r$, and 
an atom $n'$, $n' \neq n$, such that $(n,r)$ unconditionally enables 
$n'$. {\jd Moreover, if $n' \in \next (\tilde{n}, \tilde{a}, \tilde{r})$ for at least one
$\tilde{n} \neq n$ with $\tilde{a} \in P_{\tilde{n}}$ and 
 $\tilde{r} \in R_{\tilde{n}}$, 
then $\{n'\} = \next (\tilde{n}, \tilde{a}, \tilde{r})$ for some $\tilde{n} \neq n$, $\tilde{a} \in P_{\tilde{n}}$,  $\tilde{r} \in R_{\tilde{n}}$.}

\noindent {\bf Action}: 
\begin{tabular}[t]{ll}
(1) & $\outc_n \leftarrow (\outc_n \setminus \{r\}) \cup \{r'_f \mid r' \in \outc_{n'}\}$, where $r'_f$ are fresh names. \\
(2) & For all $a \in P_{n'}$, $r' \in \outc_{n'}$: $\next(n,a,r_f') \leftarrow \next(n',a ,r')$.\\
    & For all $a \in P \setminus P_{n'}$, $r' \in \outc_{n'}$: $\next(n,a,r'_f) \leftarrow \next(n,a,r)$. \\
(3) & For all $r' \in \outc_{n'}$: $\trf{n, r'_f}\leftarrow \trf{n,r}\trf{n',r'}$. \\
(4) & If $\next^{-1}(n')=\emptyset$ after (1)-(3), then remove $n'$ from $N$, where \\
    & $\next^{-1}(n') = \{ (\tilde{n},\tilde{a},\tilde{r}) \in T(N) \mid n' \in \next(\tilde{n},\tilde{a},\tilde{r}) \}$.
\end{tabular}
\end{definition}

\begin{figure}[t]
\centerline{\scalebox{0.35}{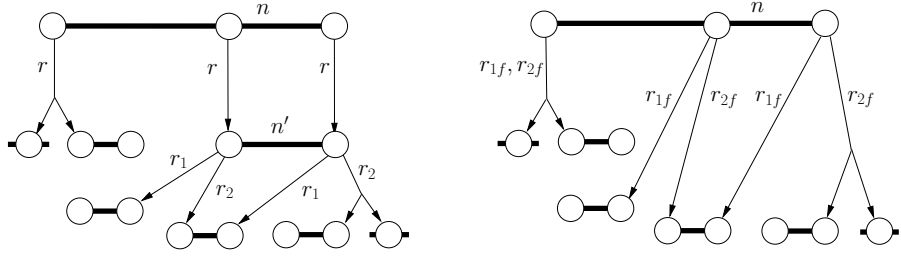}}
\caption{An application of the shortcut rule}
\label{fig:shortcut}
\end{figure}

\begin{theorem}
\label{thm:shortcound}
The shortcut rule is correct.
\end{theorem}
\begin{proof}
Let ${\cal N}_2$ be the result of applying the shortcut rule to ${\cal N}_1$.
Assume atoms $n$ and $n'$ and $r \in\outc_n$ of $n$ as in Definition \ref{def:shortcutrule}.
We say that an atom $\overline{n}$ occurs in an occurrence sequence $\sigma$ if $(\overline{n}, \overline{r})$ occurs in $\sigma$ for some $\overline{r} \in \outc_{\overline{n}}$. We call an occurrence sequence initial if it starts with the initial marking.

\noindent
{\em Claim 1.}
For each initial occurrence sequence $\sigma$ of ${\cal N}_2$, replacing all occurrences of $(n, r'_f)$ by $(n,r), (n',r')$ in $\sigma$ yields an initial occurrence sequence of ${\cal N}_1$ that leads to the same marking as $\sigma$ in ${\cal N}_2$.

\noindent
{\em Claim 2.} Each initial occurrence sequence $\sigma$ of ${\cal N}_1$ can be uniquely divided into $\sigma_1 (n,r) \sigma_1' (n',r_1) \sigma_2 (n,r) \sigma_2' (n',r_2) \ldots$ such that $(n,r)$ does not occur in the sequences $\sigma_1, \sigma_2, \ldots$ and $n'$ does not occur in $\sigma_1', \sigma_2', \ldots$. 


\noindent
{\em Claim 3.} If $\sigma = \sigma_1 (n,r) \sigma_1' (n',r_1) \sigma_2 (n,r) \sigma_2' (n',r_2) \ldots (n',r_k) \sigma_{k+1}$ (notation as in Claim 2) then
$\sigma_1 (n, r_{f 1}) \sigma'_1 \sigma_2 (n,r_{f 2}) \sigma'_2 \ldots\sigma'_k \sigma_{k+1}$ is an initial occurrence sequence of ${\cal N}_2$ leading to the same marking as $\sigma$ in ${\cal N}_1$. If there is an occurrence of $n'$ in $\sigma_i$ then a party $\tilde{a}$ of $n'$ must have been enabled by a previous occurrence of some $(\tilde{n},\tilde{a},\tilde{r})$ with $n\neq \tilde{n}$. In this case $n'$ still exists in ${\cal N}_2$.

\noindent {\em Proof:} For $i= 1, \ldots k$ we have that $\sigma_i'$ contains no atom with a party $a$ of $P_{n'}$ because $n'$ is the only atom with party $a$ enabled after the occurrence of $(n,r)$. So $\sigma_i'$ and $(n',r_i)$ can occur in arbitrary order. The sequence $(n,r)(n', r_i)$ can be replaced by the step $(n, r_{f i})$ of ${\cal N}_2$.

\noindent
{\em Claim 4.} If $\sigma = \sigma_1 (n,r) \sigma_1' (n',r_1) \sigma_2 (n,r) \sigma_2' (n',r_2) \ldots (n,r) \sigma'_{k}$ (as in Claim 2) then
$\sigma_1 (n, r_{f 1}) \sigma'_1 \sigma_2 (n,r_{f 2}) \sigma'_2 \ldots \sigma_k (n,r_{f k}) \sigma'_k$ is an initial occurrence sequence of ${\cal N}_2$ leading to the same marking as $\sigma (n', r_k)$ in ${\cal N}_1$ for an arbitrary $r_k \in \outc_{n'}$.

\noindent {\em Proof:} Since $(n,r)$ unconditionally enables $n'$, $n'$ is enabled after the last $(n,r)$ in $\sigma$, and it remains enabled because $\sigma'_k$ contains no occurrence of $n'$. So $\sigma (n', r_k)$ is also an initial occurrence sequence of ${\cal N}_1$. The claim follows using Claim 3 with $\sigma_{k+1}$ being empty.

\noindent
{\em Claim 5.} If ${\cal N}_2$ is sound then every atom $\overline{n} \in N_1$ appears in some initial occurrence sequence of ${\cal N}_1$.

\noindent {\em Proof:} If $\overline{n} \neq n'$ then this claim
follows immediately from soundness of ${\cal N}_2$ and Claim 1. Otherwise, again by its soundness, ${\cal N}_2$ has an initial occurrence sequence with
an occurrence of $n$. Clearly, the marking before the occurrence of $n$ also enables some $(n,r'_f)$, whence there is also a sequence including a step $(n,r'_f)$. Claim 1 achieves the result.

\noindent
{\em Claim 6.} If ${\cal N}_2$ is sound then every initial occurrence sequence $\sigma$ of ${\cal N}_1$ can be extended to a large step.

\noindent {\em Proof:} By Claim 3 and Claim 4, the marking reached by $\sigma$ or the marking reached by $\sigma (n,r_{f i})$ (for some $r_i \in R_{n_i})$ is also reachable in ${\cal N}_2$. Since ${\cal N}_2$ is sound, this sequence can be extended to a large step. By Claim 1, there is a corresponding occurrence sequence $\sigma'$ of ${\cal N}_1$. So either $\sigma'$ or $(n,r_{f i}) \sigma'$ extends $\sigma$ to a large step.

\noindent
{\em Claim 7.} If ${\cal N}_1$ is sound then every atom $\overline{n} \in N_2$ appears in some initial occurrence sequence of ${\cal N}_2$.

\noindent {\em Proof:}
If $\overline{n} \neq n'$ then this follows immediately from Claim 3 and Claim 4. So we only have to consider the case $\overline{n} = n'$. Then, in particular $n'$ still exists in ${\cal N}_2$. So $\next^{-1} (n')$ contains some $(\tilde{n}, \tilde{a}, \tilde{r})$ where $(\tilde{n}, \tilde{r}) \neq (n,r)$. By the guard of the shortcut rule, we have $\{n'\} = \next (\tilde{n}, \tilde{a}, \tilde{r})$ for some $\tilde{n}$, $\tilde{a}$ and $\tilde{r}$, i.e., after the occurrence of $(\tilde{n}, \tilde{r})$ only the occurrence of $n'$ can remove the token of $\tilde{a}$. Since ${\cal N}_1$ is sound, there is an initial occurrence sequence that enables $\tilde{n}$. This sequence can be extended to a large step, whence $n'$ occurs after $\tilde{n}$. Call the entire sequence $\sigma$.

Taking the division of $\sigma$ as in Claim 2, both the occurrence of $\tilde{n}$ and the subsequent occurrence of $n'$ appear in some subsequence $\sigma_i$, because the agent $\tilde{a}$ is ready to engage only in $n'$ during all $\sigma'_i$ and can thus not participate in $\tilde{n}$. The transformation of Claim 3 (Claim 4, respectively) leads to an occurrence sequence of $ {\cal N}_2$ that still includes all subsequences $\sigma'_i$ and therefore also includes an occurrence of $n'$.

\noindent
{\em Claim 8.} If ${\cal N}_1$ is sound then every initial occurrence sequence $\sigma$ of ${\cal N}_2$ can be extended to a large step.

\noindent {\em Proof:}
By Claim 1, the marking reached by $\sigma$ is also reachable in ${\cal N}_1$.
Moreover, the translated sequence has an occurrence of $n'$ after its last occurrence of $(n,r)$.
By soundness of ${\cal N}_1$, $\sigma$ can be extended by some $\sigma'$ to a large step.
Since this sequence ends with the empty marking, there is some occurrence of $n'$ after the last occurrence of $(n,r)$.
So the entire sequence $\sigma \sigma'$ can be divided into

$ \sigma_1 (n,r),(n',r_1) \sigma_2 (n,r) (n',r_2) \ldots \sigma_l \sigma_{l+1} (n,r) \sigma'_{l+1} (n', r_{l+2}) \ldots (n',r_k) \sigma_{k+1}$ (notations as in Claim 2) where $\sigma$ comprises everything up to $\sigma_l$.
We apply Claim 3 and obtain the large step of ${\cal N}_2$: 

$ \sigma_1 (n,r_{f 1}) \sigma_2 (n,r_{f 2}) \ldots \sigma_l \sigma_{l+1} (n,r_{f l+1}) \sigma'_{l+1} \ldots \sigma_{k} (n,r_{f k}) \sigma'_{k} \sigma_{k+1}$. 

\noindent
Since the subsequence until $\sigma_l$ reaches the same marking as $\sigma$, the result follows.

The proof of the theorem follows immediately from Claims 5 to 8.

\qed

\end{proof}

\noindent
{\bf \em Useless arc rule.} Consider the negotiation on the left of Figure \ref{fig:uselessarc},
in which all atoms have one outcome $r$. We have $\next(n_0,a,r) = \{n_1, n_f\}$, i.e.,
after the occurrence of $(n_0,r)$ agent $a$ is ready to engage in both $n_1$ and $n_f$. However, $a$
always engages in $n_1$, because the only large step is $(n_0,r)(n_1,r)(n_2,r)(n_f,r)$.
In other words, we can set $\next(n_0,a,n_f) = \{n_1\}$ without changing the behavior.
Intuitively, we say that the arc (more precisely, the leg of the hyperarc) leading to the $a$-port of
$n_f$ is useless. The useless arc rule identifies and removes some useless arcs.

\begin{figure}[h]
\centerline{\scalebox{0.35}{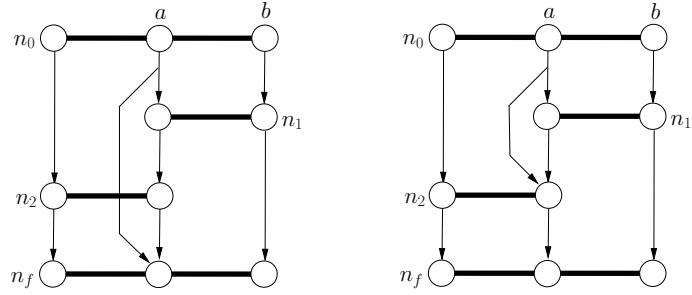}}
\caption{The useless arc rule can only be applied to the left negotiation}
\label{fig:uselessarc}
\end{figure}

\begin{definition}{Useless arc rule.}
\label{def:redu}

\noindent {\bf Guard}: There are $(n,a,r), (n,b,r) \in T(N)$ and two distinct atoms $n', n'' \in N$ 
such that $a,b \in P_{n'} \cap P_{n''}$, $n', n'' \in \next(n,a,r)$ and $\next(n,b,r) = \{n'\}$.

\noindent {\bf Action}: $\next(n,a,r) \leftarrow \next(n,a,r) \setminus \{n''\}$.
\end{definition}

The rule can be applied to the negotiation on the left of Figure \ref{fig:uselessarc}
by instantiating $n:= n_0$, $n':=n_1$, and $n'' := n_f$. 
It cannot be applied to the negotiation on the right. If we set
$n:= n_0$, $n':=n_1$, and $n'':= n_2$, then $a \notin P_{n_2}$.

\begin{theorem}
The useless arc rule is correct.
\end{theorem}
\begin{proof}

Let $\N_2$ be the result of applying the rule to $\N_1$. We adopt all notations from Definition \ref{def:redu} and claim that $\N_1$ and $\N_2$ have the same occurrence sequences.
Let $\next_1, \next_2$ be the transition functions of $\N_1$ and $\N_2$, respectively. Since $\next_2(n,a,r) \subseteq \next_1(n,a,r)$ for every $(n,a,r) \in T(N)$, every occurrence sequence of $\N_2$ is also an occurrence sequence of $\N_1$.
It remains to prove that every initial occurrence sequence $\sigma$ of $\N_1$ is also an initial occurrence sequence
of $\N_2$.

Assume the contrary, i.e., $\sigma = \sigma_1 \sigma_2$ such that $\sigma_1$ is an initial occurrence sequence of both $\N_1 $ and $\N_2$ whereas the first step of $\sigma_2$ is only possible in $\N_1$. Since both negotiations only differ w.r.t.\ $\next (n,a,r)$, this step must be an occurrence of agent $n''$
occurring after an occurrence of $(n,r)$ such that no other atom of $\next (n,a,r)$ occurred in between. This holds in particular for $n' \in \next (n,a,r)$.
But, since $\next (n,b,r) = \{n'\}$ and $b$ participates both in $n'$ and in $n''$, $n'$ must occur before $n''$ after $(n,r)$ -- a contradiction.
This concludes the proof showing that $\N$ and $\N'$ have the
same occurrence sequences.

A first immediate consequence is that an atom can occur in $\N_1$ if{}f it can occur in $\N_2$. It remains to show that every occurrence sequence of $\N_1$ can be extended to a large step if{}f the same holds for $\N_2$. To this end, since both negotiations have the same occurrence sequences, we only have to show that an occurrence sequence is a large step of $\N_1$ if{}f it is a large step of $\N_2$. Respective markings in $\N_1$ and $\N_2$ reached by the same occurrence sequence differ only with respect to agent $a$: after the occurrence of $(n,r)$, $\vx (a)$ contains $n'$ and $n''$ in $\N_1$ and $\vx (a)$ contains $n'$ in $\N_2$. Large steps, however, lead to the marking satisfying $\vx (a) = \emptyset$ and are hence identical for both negotiations.

\qed

\end{proof}

\section{(Weakly) Deterministic Negotiations}
\label{sec:determ}

We introduce weakly deterministic and deterministic negotiations.

\begin{definition}
  An agent $a \in \agents$ is {\em deterministic} if for every
    $(n,a,r) \in T(N)$ such that $n \neq n_f$ there exists one atom
    $n'$ such that $\next(n,a,r) = \{n'\}$. 
    
The negotiation $\N$  is {\em weakly deterministic} if for every $(n,a,r) \in T(N)$ there is a deterministic
agent $b$ that is a party of every atom in $\next(n,a,r)$, i.e., $b \in P_{n'}$ for every $n' \in \next(n,a,r)$.
It is {\em deterministic} if all its agents are deterministic.
\end{definition}
Graphically, an agent $a$ is deterministic if no proper hyperarc leaves any port of $a$.
Consider the negotiations of Figure \ref{fig:dneg}. In the acyclic negotiation both
Father and Daughter are deterministic, while Mother is not. In the ping-pong negotiation
only Daughter is deterministic. Both negotiations are weakly deterministic, because
Daughter participates in all atoms, and so can be always chosen as the party $b$ required by the 
definition. 
Observe that the notion of deterministic agent does not refer to the behavior of  atoms,
which is intrinsically nondeterministic with respect to its possible outcomes and even to its state transformations. Rather, it refers to the {\em composition} of negotiations: For each atom $n$, the 
next atom of a deterministic agent is completely determined by the outcome of $n$.

Weakly deterministic negotiations have a natural semantic justification.
Consider a negotiation with two agents $a,b$ and three atoms
$\{n_0,n_1,n_f\}$.  All atoms have the same parties $a,b$ and one
outcome $r$, such that $\next(n_0,a,r) = \{n_1,n_f\} = \next(n_0,b,r)$
and $\next(n_1,a,r) = \{n_f\} = \next(n_1,b,r)$. After the occurrence of
$(n_0, r)$ the parties $a$ and $b$ are ready to engage in both $n_1$ and
$n_f$, and so which of them occurs requires a ``meta-negotiation''
between $a$ and $b$. This meta-negotiation, however, is not part of the model 
and, more importantly, it can be difficult to implement, since it requires 
to break a symmetry. In a weakly deterministic negotiation
this situation cannot happen. If $\next(n,a,r)$ contains more than one atom,
then some deterministic agent $b$ is a part of all atoms in $\next(n,a,r)$.
If some atom of $\next(n,a,r)$ becomes enabled, say $n'$, then 
because agent $b$ is ready to engage in it, and, since $b$ is deterministic, 
$b$ is not ready to engage in any other atom. So $n'$ is the only enabled atom of 
$\next(n,a,r)$, and $n'$ is the negotiation that $a$ will engage in next. This
is very easy to implement: $b$ just sends a message to $a$ telling her that she
should commit to $n'$.

Notice that, for  deterministic negotiations, the second part of the guard of the shortcut rule is always satisfied. Using the notation of the shortcut rule,
this condition requires that the atom $n'$ is the only atom in $\next (\tilde{n},\tilde{a}, \tilde{r})$ for some $(\tilde{n}, \tilde{r})$, provided $n'$ is not only in $\next (n,a,r)$ for some $a \in P_n$ (i.e., provided $n'$ is not removed by the application of the rule). This clearly holds if $\tilde{n}$ is deterministic, as all agents are deterministic in deterministic negotiations.

In the next section we show that, on top of their semantic justification, 
weakly deterministic and deterministic negotiations are also interesting from an analysis point
of view. We prove that 
the shortcut and useless arc rules are complete for acyclic, weakly deterministic negotiations, 
which of course implies that the same rules plus the merge are complete, too. 
A second result  proves that a polynomial number of applications of the merge
and shortcut rules suffices to summarize any sound deterministic
acyclic negotiation (the useless arc rule is irrelevant for deterministic 
negotiations).

\section{Completeness and complexity}
\label{sec:compcomp}

We start with the completeness result for the weakly deterministic case.

\begin{theorem}
\label{thm:shortcomp}
The shortcut and useless arc rules are complete for acyclic, weakly deterministic negotiations. 
\end{theorem}
\begin{proof}
Let ${\cal N}$ be a sound, acyclic, and weakly deterministic negotiation.

The proof has two parts:\\[0.2cm]
\noindent (1) If ${\cal N}$ has more than one atom, then the shortcut rule or the useless arc rule
can be applied to it.

Since ${\cal N}$ is acyclic, its graph generates a partial order on atoms in the obvious way ($n < n'$ if there is a path from $n$ to $n'$).
Clearly $n_0$ is the unique minimal element.
We choose an arbitrary linearisation of this partial order.
Since ${\cal N}$ has more than one atom, this linearisation begins with $n_0$ and has some second element, say $n_1$.
Since ${\cal N}$ is sound, some occurrence sequence begins with an occurrence of $n_0$ and a subsequent occurrence of $n_1$.
So $n_0$ has an outcome $r_0$ such that $n_1 \in \next(n_0,a,r_0)$
for every party $a$ of $n_1$.

Consider two cases:
\begin{itemize}
\item $\{n_1\} = \next(n_0,a,r_0)$ for every party $a$ of $n_1$.\\
Then $(n_0, r_0)$ unconditionally enables $n_1$. Moreover, there are
no $\tilde{n}$, $\tilde{a}$ and $\tilde{r}$ such that $\tilde{n} \neq n_0$ and
$\next (\tilde{n}, \tilde{a}, \tilde{r}) = \{n_1\}$ because otherwise, according to the above defined partial order, $n_0 < \tilde{n} < n_1$ and so $\tilde{n}$ would be between $n_0$ and $n_1$ in every linearisation. So the shortcut rule can be applied.

\item $\{n_1\} \neq \next(n_0,a,r_0)$ for some party $a$ of $n_1$.\\
Then $n_1, n_2 \in \next(n_0,a,r_0)$ for some atom $n_2 \neq n_1$.
Since $\N$ is weakly deterministic, there is a deterministic agent $b$ that is a party
of every atom in $\next(n_0,a,r_0)$, in particular of both $n_1$ and $n_2$. Since $b$
is deterministic and $n_1 \in \next (n_0, b, r)$, we have $\next(n_0,b,r_0) = \{n_1\}$, and so
the useless arc rule is applicable.
\end{itemize}

\noindent (2) The shortcut and useless arc rules cannot be applied infinitely often. \\
Let ${\cal N}_2$ be the result of applying any of the two rules to a
negotiation ${\cal N}_1$. For every large step $\sigma'$ of ${\cal N}_2$, let $\phi'(\sigma')$
be defined similarly to Theorem \ref{thm:shortcound}: if the useless arc rule has been applied, then $\phi'(\sigma')=\sigma'$; if the shortcut rule has been applied
to atoms $n, n'$, then let $\phi'(\sigma')$ be the result of replacing every occurrence of $(n,r_f')$ by
the sequence $(n,r)(n',r')$. We have $|\sigma'| \leq |\phi(\sigma')|$. Moreover,
if the shortcut rule is applied, then $|\sigma'| < |\sigma|$ for at least one large step $\sigma'$,
indeed for all large steps containing $(n',r_f')$.
Since the set of large steps of an acyclic negotiation is finite, and the length of
every large step is not negative, no infinite sequence of
applications of the rules can contain infinitely many applications of the shortcut rule.
So any infinite sequence of applications must, from some point on, only apply the useless arc rule.
But this is also not possible, since the rule reduces the number of arcs.

By (2) every maximal sequence of applications of the rule terminates.
By (1) it terminates with a negotiation containing one single atom.
\qed
\end{proof}

Next we prove that a {\em polynomial number} of applications of the merge and shortcut
rules suffice to summarize any sound {\em  deterministic} acyclic negotiation
(SDAN). For this we have to follow a strategy in the application of the
shortcut rule.

\begin{definition}
The {\em deterministic shortcut rule}, or {\em d-shortcut} rule, is the result of
adding to the guard of the shortcut rule a new condition: (3) $n'$ has at most one
outcome (the actions of the shortcut and d-shortcut rules coincide).
\end{definition}

We say that a SDAN is {\em irreducible} if neither the merge nor the d-shortcut rule 
can be applied to it. In the rest of the section we prove that irreducible SDANs
are necessarily atomic, i.e., consist of a single atom. The proof, which is rather involved, proceeds in three steps. First, we prove a technical
lemma showing that SDANs can be reduced so that all agents participate in every atom with more 
than one outcome. Then we use this result to prove that, loosely speaking, every
SDAN can be reduced to a ``replication'' of a negotiation with only one agent: if 
$\next(n,a,r)=\{n'\}$ for some agent $a$, then $\next(n,a,r)=\{n'\}$ for {\em every} agent $a$.
In the third step, we show that replications can be reduced to atomic negotiations. Finally,
we analyze the number of rule applications needed in each of these three steps, and conclude.

\begin{lemma}
\label{lemma1} 
Let ${\cal N}$ be an irreducible SDAN and 
let $n \neq n_f$ be an atom of ${\cal N}$ with more than one outcome.
Then every agent participates in $n$. 
\end{lemma}
\begin{proof}
We proceed in two steps.\\[0.2cm]
\noindent (a) The atom $n$ has an outcome $r$ such that: either $(n,r)$ unconditionally enables $n_f$, or $(n,r)$ unconditionally enables some atom with more than one outcome.

This is the core of the proof. We first claim: if some outcome $(n,r)$ unconditionally
enables some atom, then (a) holds.
Indeed: if $(n,r)$ unconditionally enables some atom $n'$, then either $n'=n_f$ or $n'$
has more than one outcome, because otherwise the d-shortcut rule can be applied to
$n$ and $n'$, contradicting the irreducibility of $\N$. This proves the claim.

It remains to prove that some outcome $(n,r)$ unconditionally enables some atom. For this,
we assume the contrary, and prove that ${\cal N}$ contains a cycle, contradicting the
hypothesis.

Since the merge rule is not applicable to ${\cal N}$,
$n$ has two outcomes $r_1, r_2$ such that $\next(n,a,r_1) \neq \next(n,a,r_2)$
for some agent $a$. We proceed in three steps.

(a1) For every reachable marking $\vx$ that enables $n$ there is a
sequence $\sigma$ such that $\vx \by{(n,r_1) \, \sigma} \vx_1$ and $\vx \by{(n,r_2) \, \sigma} \vx_2$ for some markings $\vx_1, \vx_2$, and the sets $N_1$ and $N_2$ of
atoms enabled by $\vx_1, \vx_2$ are nonempty and disjoint.\\
Let $\sigma$ be a longest occurrence sequence such that $\vx \by{(n,r_1)\, \sigma} \vx_1$
and $\vx \by{(n,r_2)\, \sigma} \vx_2$ for some markings $\vx_1, \vx_2$ (notice that
$\sigma$ exists, because all occurrence sequences of $\N$ are finite by acyclicity).
We have $N_1 \cap N_2 = \emptyset$, because otherwise we can extend $\sigma$ with the
occurrence of any atom enabled by both markings.
We prove that, furthermore, $N_1 \neq \emptyset \neq N_2$. Assume w.l.o.g. $N_1 = \emptyset$.
Then, since ${\cal N}$ is sound, we have $\vx_1 = \vx_f$, which means that the last
step of $\sigma$ is of the form $(n_f, r_f)$. So $\vx_2$ is also a marking obtained after
the occurrence of $(n_f,r_f)$. Since every agent participates in $n_f$ and $\next(n_f,a,r_f)=\emptyset$ for every
agent $a$ and outcome $r_f$, we also have $\vx_2 = \vx_f$. So
$\vx \by{(n,r_1)} \vx_1' \by{\sigma} \vx_f$ and $\vx \by{(n,r_1)} \vx_2' \by{\sigma} \vx_f$, which implies $\vx_1' = \vx_2'$. Since $\N$ is deterministic, we then have $\next(n,a,r_1) = \next(n,a,r_2)$, contradicting the hypothesis.

(a2) For every $n_1 \in N_1$ there is a path leading from some $n_2 \in N_2$ to $n_1$, and
for every $n_2 \in N_2$ there is a path leading from some $n_1 \in N_1$ to $n_2$.\\
By symmetry it suffices to prove the first part. Since $N_1$ and $N_2$ are disjoint,
$n_1$ is enabled at $\vx_1$ but not at $\vx_2$. Moreover, since $\N$ is acyclic, every
atom can occur at most once in an occurrence sequence, and so neither $n_1$ nor $n_2$ appear
in $\sigma$. Since, furthermore, the sequences $(n,r_1)\, \sigma$ and $(n,r_2)\, \sigma$ only differ in their first element, there is an agent $a$ such that
$\next(n,a,r_1) = \{n_1\}$ and $\next(n,a,r_2)=\{n_2'\}$ for some $n_2' \neq n_1$ ($n_2'$ is not necessarily the $n_2 \in N_2$ we are looking for).
So we have $\vx_1(a) = \{n_1\}$ and $\vx_2(a)= \{n_2'\}$ (see Figure \ref{fig:polyproof}).

\begin{figure}[h]
\centerline{\scalebox{0.45}{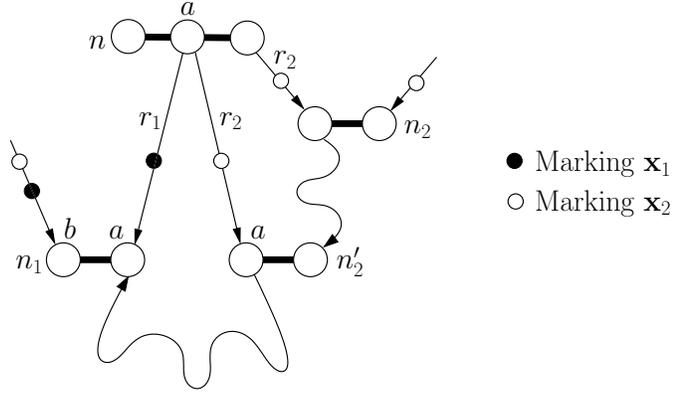}}
\caption{Illustration of the proof of Lemma \ref{lemma1}.}
\label{fig:polyproof}
\end{figure}

We first show that there is a path from $n_2'$ to $n_1$.
By assumption, no outcome of $n$ unconditionally enables any atom, and so
$(n,r_1)$ does not unconditionally enable $n_1$.
So $n_1$ has a party $b \neq a$ such that either $b$ is not a party of $n$ or
$\next(n,b,r_1) \neq n_1$. Since $\vx_1$ enables $n_1$ we have $\vx_1(b)= \{n_1\}$, and since
$b$ is not a party of $n$ or $\next(n,b,r_1) \neq n_1$, we have $\vx_2(b)=\{n_1\}$ as well.
Since $\vx_2(b)=n_1$, and ${\cal N}$ is a SDAN, there is an occurrence sequence
$\tau$ such that $\vx_2 \by{\tau} \vx_2'$ and $\vx_2'$ enables $n_1$ (intuitively,
the white token on the arc to the $b$-port can only leave the arc through the occurrence of $n_1$). Since $\vx_2(a)=\{n_2'\} \neq \{n_1\}$, there is a path from $n_2'$ to $n_1$ (intuitively,
the white token on the arc leading to the $a$-port of $n_2'$ has to travel
to some arc leading to the $a$-port of $n_1$, and by determinism it can only do so
through a path of $a$-ports that crosses $n_2'$).

We now prove that there is a path from some $n_2 \in N_2$ to $n_2'$. If $n_2'$ is enabled at $\vx_2$, then $n_2' \in N_2$ and we are done. If $n_2'$ is not enabled at $\vx_2$ (as in the figure) then, since $\vx_2(a)= \{n_2'\}$ and ${\cal N}$ is a SDAN,
there is a sequence $\tau$ such
that $\vx_2 \by{\tau} \vx_2'$ and $\vx_2'$ enables $n_2'$ (again, by soundness
the white token on the arc to the $a$-port of $n_2'$ can eventually leave the arc to move towards $n_f$, and by determinacy it can only leave the arc through the occurrence of $n_2'$).
Since $N_2$ is the set of transitions enabled at $\vx_2$, we have $\tau = (n_2, r) \, \tau'$ for some $n_2 \in N_2$. So some subword of $\tau$ is a path from some transition of $N_2$ to $n_2'$.

(a3) ${\cal N}$ contains a cycle. \\
Follows immediately from (a2) and the finiteness of $N_1$ and $N_2$.\\

\noindent (b) Every agent participates in $n$.\\
By repeated application of (a) we find a chain $(n_1, r_1) \ldots (n_k, r_k)$ such that
$n_1 = n$, $n_k = n_f$, and $(n_i, r_i)$ unconditionally enables $n_{i+1}$ for every $1 \leq i \leq k-1$. By the definition of unconditionally enabled we have
$P_1 \supseteq P_2 \supseteq \cdots \supseteq P_k = P_f$. Since $P_f = \agents$,
we obtain $P_1 = \agents$.
\end{proof}

\begin{lemma}
\label{lemma2}
Let ${\cal N}$ be an irreducible SDAN. Every agent participates in every atom,
and for every atom $n \neq n_f$ and every outcome $r$ there is an atom $n'$ satisfying 
$\next(n,a,r)=\{n'\}$ for every agent $a$. 
\end{lemma}
\begin{proof}
We first show that every agent participates in every atom. By Lemma~\ref{lemma1},
it suffices to prove that every atom $n \neq n_f$ has more than one outcome.
Assume the contrary, i.e., some atom different from $n_f$ has only one outcome.
Since, by soundness, every atom can occur, there is an occurrence sequence
$(n_0, r_0) (n_1, r_1) \cdots (n_k, r_k)$ such that $n_k$ has only one outcome
and all of $n_0, \ldots, n_{k-1}$ have more than one outcome. By Lemma \ref{lemma1}, 
all agents participate in all of $n_0, n_1, n_{k-1}$. It follows that $(n_i, r_i)$ 
unconditionally enables $(n_{i+1}, r_{i+1})$ for every $0 \leq i \leq k-1$. 
In particular, $(n_{k-1}, r_{k-1})$ unconditionally enables $(n_k, r_k)$. But 
then, since $n_k$ only has one outcome, the d-shortcut rule can be applied to 
$n_{k-1}, n$, contradicting the hypothesis that ${\cal N}$ is irreducible.

For the second part, assume there is an atom $n \neq n_f$, an outcome $r$ of $n$, and two
distinct agents $a_1, a_2$ such that $\next(n,a_1,r)= \{n_1\} \neq \{n_2\} = \next(n,a_2,r)$. By the first part,
every agent participates in $n$, $n_1$ and $n_2$. Since ${\cal N}$ is sound, some reachable
marking $\vx$ enables $n$. Moreover, since all agents participate in $n$, and ${\cal N}$ is
deterministic, the marking $\vx$ only enables $n$. Let $\vx'$ be the marking given by
$\vx \by{(n,r)} \vx'$. Since $a_1$ participates in all atoms, no atom different from $n_1$ can be enabled
at $\vx'$. Symmetrically, no atom different from $n_2$ can be enabled
at $\vx'$. So $\vx'$ does not enable any atom, contradicting that ${\cal N}$ is sound. \qed
\end{proof}

\begin{theorem}
\label{irredtheo}
Let ${\cal N}$ be an irreducible SDAN. Then ${\cal N}$ contains only one atom. 
\end{theorem}
\begin{proof}(Sketch)
Assume ${\cal N}$ contains more than one atom. Fore every atom $n \neq n_f$,
let $l(n)$ be the length of the longest path from $n$ to $n_f$ in the graph of $\N$. Let
$n_{\min}$ be any atom such that $l(n_{\min})$ is minimal, and let $r$ be an arbitrary outcome of $n_{\min}$. By Lemma \ref{lemma2} there is an atom
$n'$ such that $\next(n_{\min},a,r)=\{n'\}$ for every agent $a$. If $n' \neq n_f$ then by
acyclicity we have $l(n') < l(n_{\min})$, contradicting the minimality of $n_{\min}$.
So we have $\next(n_{\min},a,r)=\{n'\}$ for every outcome $r$ of $n_{\min}$ and every agent $a$.
If $n_{\min}$ has more than one outcome, then the merge rule is applicable. If $n_{\min}$ has one outcome, then, since it unconditionally enables $n_f$, the d-shortcut rule is applicable. 
In both cases we get a contradiction to irreducibility. \qed
\end{proof}

\begin{definition}
For every atom $n$ and outcome $r$, let ${\it shoc}(n,r)$ be the length of a shortest maximal 
occurrence sequence containing $(n,r)$ minus $1$, and let ${\it Shoc}({\cal N}) = \sum_{n \in N, r \in R} {\it shoc}(n,r)$. Finally, let
${\it Out}({\cal N}) = \sum_{(P,R,\delta) \in N \setminus \{ n_f\}} |R|$ be the total number of outcomes of ${\cal N}$, excluding those of the final atom.
\end{definition}

Notice that if ${\cal N}$ has $K$ atoms then ${\it shoc}(n,r) \leq K-1$ holds for every atom $n$ and outcome $r$.
Further, if $K=1$ then ${\it Shoc}({\cal N})= 0 = {\it Out}({\cal N})$.

\begin{theorem}
\label{thm:polcomp}
Every SDAN ${\cal N}$ can be completely reduced by means of
${\it Out}({\cal N})$ applications of the merge rule and ${\it Shoc}({\cal N})$ applications of the 
d-shortcut rule.
\end{theorem}
\begin{proof}
Let  ${\cal N}$ and ${\cal N}'$ be negotiations such that
${\cal N}'$ is obtained from ${\cal N}$ by means of the merge or the d-shortcut rule. 
For the merge rule we have ${\it Out}({\cal N}') < {\it Out}({\cal N})$ and 
${\it Shoc}({\cal N}') \leq {\it Out}({\cal N})$
because the rule reduces the number of outcomes by one. For the d-shortcut rule we have 
${\it Out}({\cal N}') = {\it Out}({\cal N})$
because if it is applied to pairs $n, n'$ such that $n'$ has one single outcome,
and ${\it Shoc}({\cal N}') <  {\it Shoc}({\cal N})$, because ${\it Shoc}(n,r_f') < {\it Shoc}(n,r)$. \qed
\end{proof}

\section{Conclusions}
\label{sec:conc}

We have introduced negotiations, a formal model of concurrency with
 negotiation atoms as primitive. 
We have defined and studied two important analysis problems: soundness, which coincides 
with the soundness notion for workflow nets, and the new {\em summarization problem}. 
We have provided a complete set of rules for sound acyclic, weakly deterministic negotiations,
and we have shown that the rules allow one to compute a summary of a sound
deterministic negotiations in polynomial time. 

Several open questions deserve further study.
Our results show that  summarization  can be solved in polynomial
time for deterministic, acyclic negotiations, and is co-NP-hard for arbitrary acyclic negotiations. 
The precise complexity of the weak deterministic case is still open. We 
are currently working on a  generalization
of Rule 2.2. of Section \ref{subsec:sumalg}, such that we can completely reduce in polynomial time {\em arbitrary} deterministic negotiations, even if they contain cycles. \\

\noindent
{\bf \em Related work.} 
Previous work on Petri net analysis by means of reductions has already 
been discussed in the Introduction.

A number of papers have modelled specific distributed negotiation protocols with the help of Petri nets or process 
calculi (see  \cite{SalaunFC04,JiTL05,BacarinMMA11}).
However, these papers do not address the issue of negotiation as concurrency primitive. 

The feature of summarizing parts of a negotiation to  single negotiation atoms has several analogies in Petri net theory, among these the concept of zero-safe Petri nets.
By abstracting from reachable markings which mark distinguished ``zero-places'', transactions can be modelled by zero-safe Petri nets \cite{DBLP:conf/apn/BruniM00}. Reference \cite{DBLP:conf/ac/BruniMM03} extends this concept to reconfigurable, dynamic high-level Petri nets.

A  related line of research studies global types and session types to model multi-party sessions \cite{DBLP:conf/popl/HondaYC08}. See  \cite{DBLP:journals/corr/abs-1203-0780} for an overview that also covers choreography-based approaches for web services. This research emphasises communication aspects in the formal setting of mobile processes. Thus, the aim differs from our aim. However, it might be worth trying to combine the two approaches. 

Finally, the graphical representation of negotiations was partly inspired by the BIP component framework \cite{DBLP:journals/software/BasuBBCJNS11,DBLP:journals/tc/BliudzeS08}, where a set of sequential components (i.e., the agents) interact by synchronizing on certain actions (i.e., the atoms).\\

\noindent
{\bf \em Acknowledgement.} 
We thank the reviewers, in particular for their hints to related papers. We 
also thank Stephan Barth, Eike Best, and Jan Kretinsky
for very helpful discussions.

\bibliographystyle{is-abbrv} 
\bibliography{references}

\newcommand{\noopsort}[1]{} \newcommand{\singleletter}[1]{#1}
\begin{thebibliography}{10}
\ifx \showCODEN  \undefined \def \showCODEN #1{CODEN #1}  \fi
\ifx \showISBN   \undefined \def \showISBN  #1{ISBN #1}   \fi
\ifx \showISSN   \undefined \def \showISSN  #1{ISSN #1}   \fi
\ifx \showLCCN   \undefined \def \showLCCN  #1{LCCN #1}   \fi
\ifx \showPRICE  \undefined \def \showPRICE #1{#1}        \fi
\ifx \showURL    \undefined \def \showURL {URL }          \fi
\ifx \path       \undefined \input path.sty               \fi
\ifx \ifshowURL \undefined
     \newif \ifshowURL
     \showURLtrue
\fi

\bibitem{aalst}
W.~M.~P. {\noopsort{Aalst}}van~der Aalst.
\newblock The application of {P}etri nets to workflow management.
\newblock {\em J. Circuits, Syst. and Comput.}, 08\penalty0 (01):\penalty0
  21--66, 1998.

\bibitem{DBLP:journals/fac/AalstHHSVVW11}
W.~M.~P. {\noopsort{Aalst}}van~der Aalst, K.~M. van Hee, A.~H.~M. ter Hofstede,
  N.~Sidorova, H.~M.~W. Verbeek, M.~Voorhoeve, and M.~T. Wynn.
\newblock Soundness of workflow nets: classi\-fication, decidability, and
  analysis.
\newblock {\em Formal Asp. Comput.}, 23\penalty0 (3):\penalty0 333--363, 2011.

\bibitem{BacarinMMA11}
E.~Bacarin, E.~R.~M. Madeira, C.~B. Medeiros, and W.~M.~P. van~der Aalst.
\newblock {\it SpiCa}'s multi-party negotiation protocol: Implementation using
  {YAWL}.
\newblock {\em Int. J. Cooperative Inf. Syst.}, 20\penalty0 (3):\penalty0
  221--259, 2011.

\bibitem{DBLP:journals/software/BasuBBCJNS11}
A.~Basu, S.~Bensalem, M.~Bozga, J.~Combaz, M.~Jaber, T.-H. Nguyen, and
  J.~Sifakis.
\newblock Rigorous component-based system design using the {BIP} framework.
\newblock {\em IEEE Software}, 28\penalty0 (3):\penalty0 41--48, 2011.

\bibitem{DBLP:conf/ac/Berthelot86}
G.~Berthelot.
\newblock Transformations and decompositions of nets.
\newblock In W.~Brauer, W.~Reisig, and G.~Rozenberg, editors, {\em Advances in
  Petri Nets}, volume 254 of {\em LNCS}, pages 359--376. Springer, 1986.

\bibitem{DBLP:journals/tc/BliudzeS08}
S.~Bliudze and J.~Sifakis.
\newblock The algebra of connectors - structuring interaction in {BIP}.
\newblock {\em IEEE Trans. Computers}, 57\penalty0 (10):\penalty0 1315--1330,
  2008.

\bibitem{DBLP:conf/ac/BruniMM03}
R.~Bruni, H.~C. Melgratti, and U.~Montanari.
\newblock Extending the zero-safe approach to coloured, reconfigurable and
  dynamic nets.
\newblock In J.~Desel, W.~Reisig, and G.~Rozenberg, editors, {\em Lectures on
  Concurrency and Petri Nets}, volume 3098 of {\em Lecture Notes in Computer
  Science}, pages 291--327. Springer, 2003.

\bibitem{DBLP:conf/apn/BruniM00}
R.~Bruni and U.~Montanari.
\newblock Executing transactions in zero-safe nets.
\newblock In M.~Nielsen and D.~Simpson, editors, {\em Application and Theory of
  Petri Nets}, pages 83--102, 2000.

\bibitem{DBLP:journals/corr/abs-1203-0780}
G.~Castagna, M.~Dezani-Ciancaglini, and L.~Padovani.
\newblock On global types and multi-party session.
\newblock {\em Logical Methods in Computer Science}, 8\penalty0 (1), 2012.

\bibitem{DBLP:conf/concur/Desel90}
J.~Desel.
\newblock Reduction and design of well-behaved concurrent systems.
\newblock In J.~C.~M. Baeten and J.~W. Klop, editors, {\em CONCUR}, volume 458
  of {\em Lecture Notes in Computer Science}, pages 166--181. Springer, 1990.

\bibitem{Desel:1995:FCP:207572}
J.~Desel and J.~Esparza.
\newblock {\em Free choice Petri nets}.
\newblock Cambridge University Press, New York, NY, USA, 1995.

\bibitem{DBLP:conf/caise/DongenAV05}
B.~F. {\noopsort{Dongen}}van~Dongen, W.~M.~P. van~der Aalst, and H.~M.~W.
  Verbeek.
\newblock Verification of {EPC}s: Using reduction rules and {P}etri nets.
\newblock In O.~Pastor and J.~F. e~Cunha, editors, {\em CAiSE}, volume 3520 of
  {\em LNCS}, pages 372--386. Springer, 2005.

\bibitem{DBLP:conf/ac/Esparza96}
J.~Esparza.
\newblock Decidability and complexity of {P}etri net problems - an
  introduction.
\newblock In {\em Petri Nets}, volume 1491 of {\em LNCS}, pages 374--428.
  Springer, 1996.

\bibitem{Concur2013}
J.~Esparza and J.~Desel.
\newblock On negotiation as concurency primitive, 2013.
\newblock To appear in the Proceedings of CONCUR 2013.

\bibitem{DBLP:journals/tcs/GenrichT84}
H.~J. Genrich and P.~S. Thiagarajan.
\newblock A theory of bipolar synchronization schemes.
\newblock {\em Theor. Comput. Sci.}, 30:\penalty0 241--318, 1984.

\bibitem{DBLP:conf/apn/Haddad88}
S.~Haddad.
\newblock A reduction theory for coloured nets.
\newblock In G.~Rozenberg, editor, {\em Advances in Petri Nets}, volume 424 of
  {\em LNCS}, pages 209--235. Springer, 1988.

\bibitem{DBLP:journals/ppl/HaddadP06}
S.~Haddad and J.-F. Pradat-Peyre.
\newblock New efficient {P}etri nets reductions for parallel programs
  verification.
\newblock {\em Parallel Processing Letters}, 16\penalty0 (1):\penalty0
  101--116, 2006.

\bibitem{DBLP:conf/popl/HondaYC08}
K.~Honda, N.~Yoshida, and M.~Carbone.
\newblock Multiparty asynchronous session types.
\newblock In G.~C. Necula and P.~Wadler, editors, {\em POPL}, pages 273--284.
  ACM, 2008.

\bibitem{HUM}
J.~E. Hopcroft, R.~Motwani, and J.~D. Ullman.
\newblock {\em Introduction to Automata Theory, Languages, and Computation (3rd
  Edition)}.
\newblock Addison-Wesley Longman Publishing Co., Inc., Boston, MA, USA, 2006.

\bibitem{JiTL05}
S.~Ji, Q.~Tian, and Y.~Liang.
\newblock A {P}etri-net-based modeling framework for automated negotiation
  protocols in electronic commerce.
\newblock In D.~Lukose and Z.~Shi, editors, {\em PRIMA}, volume 4078 of {\em
  LNCS}, pages 324--336. Springer, 2005.

\bibitem{PapadimitriouY82}
C.~H. Papadimitriou and M.~Yannakakis.
\newblock The complexity of facets (and some facets of complexity).
\newblock In H.~R. Lewis, B.~B. Simons, W.~A. Burkhard, and L.~H. Landweber,
  editors, {\em STOC}, pages 255--260. ACM, 1982.
\newblock \showISBN{0-89791-070-2}.

\bibitem{SalaunFC04}
G.~Sala{\"u}n, A.~Ferrara, and A.~Chirichiello.
\newblock Negotiation among web services using {LOTOS/CADP}.
\newblock In L.-J. Zhang, editor, {\em ECOWS}, volume 3250 of {\em LNCS}, pages
  198--212. Springer, 2004.

\bibitem{DBLP:journals/jcss/VerbeekWAH10}
H.~M.~W. Verbeek, M.~T. Wynn, W.~M.~P. van~der Aalst, and A.~H.~M. ter
  Hofstede.
\newblock Reduction rules for reset/inhibitor nets.
\newblock {\em J. Comput. Syst. Sci.}, 76\penalty0 (2):\penalty0 125--143,
  2010.

\end{thebibliography}
\end{document}